\documentclass{article}

\usepackage{PRIMEarxiv}

\usepackage[utf8]{inputenc} % allow utf-8 input
\usepackage[T1]{fontenc}    % use 8-bit T1 fonts
\usepackage{appendix} 
\usepackage{hyperref}       % hyperlinks
\usepackage{url}            % simple URL typesetting
\usepackage{booktabs}       % professional-quality tables
\usepackage{amsfonts}       % blackboard math symbols
\usepackage{amsmath}
\usepackage{cite}
\usepackage{nicefrac}       % compact symbols for 1/2, etc.
\usepackage{microtype}      % microtypography
\usepackage{lipsum}
\usepackage{fancyhdr}       % header
\usepackage{graphicx}       % graphics
\graphicspath{{media/}}     % organize your images and other figures under media/ folder
\usepackage{makecell}
\allowdisplaybreaks
\usepackage{enumerate}
\newtheorem{remark}{Remark}
\newtheorem{theorem}{Theorem}
\newtheorem{lemma}{Lemma}
\newtheorem{proof}{Proof}

\title{Cloud Control of Connected Vehicle under Bi-directional Time-varying delay: An Application of Predictor-observer Structured Controller
}

\author{
  Ji-An Pan \\
  School of Vehicle and Mobility \\
  Tsinghua University \\
  Beijing\\
  \texttt{pja17@mails.tsinghua.edu.cn} \\
   \And
  Qing Xu, Keqiang Li, Jianqiang Wang \\
  School of Vehicle and Mobility \\
  Tsinghua University\\
  Beijing\\
  \texttt{\{qingxu,likq,wjqlws\}@tsinghua.edu.cn} \\  
  \AND
  Chunying Yang \\
  School of Vehicle and Mobility \\
  Tsinghua University\\
  Beijing\\
  \texttt{SY1913312@buaa.edu.cn}
}

\begin{document}
\maketitle

\begin{abstract}
This article is devoted to addressing the cloud control of connected vehicles, specifically focusing on analyzing the effect of bi-directional communication-induced delays. To mitigate the adverse effects of such delays, a novel predictor-observer structured controller is proposed which compensate for both measurable output delays and unmeasurable, yet bounded, input delays simultaneously. The study begins by novelly  constructing an equivalent delay-free inter-connected system model that incorporates the Predictor-Observer controller, considering certain delay boundaries and model uncertainties. Subsequently, a stability analysis is conducted to assess the system's robustness under these conditions.
Next, the connected vehicle lateral control scenario is built which contain high-fidelity vehicle dynamic model. The results demonstrate the controller's ability to accurately predict the system states, even under time-varying bi-directional delays.
Finally, the proposed method is deployed in a real connected vehicle lateral control system. Comparative tests with a conventional linear feedback controller showcase significantly improved control performance under dominant bi-directional delay conditions, affirming the superiority of the proposed method against the delay.
\end{abstract}

% keywords can be removed

\section{Introduction}

The advent of fast-growing connection technology has led to the integration of cloud-driven functionalities into autonomous driving systems. ``Cloud control of Intelligent and Connected Vehicle (ICV)\cite{chang2020,chang2020jienengxuebao}" is one of the road maps proposed under that background. The fundamental idea behind the road map entails the establishment of a dedicated transportation-exclusive public cloud platform, capable of off-loading select on-board perception, decision-making, and even control functions to the cloud system. This approach empowers the cloud control platform to offer various levels of driving assistance services\cite{ITU_NDA,bian2019cooperation} to ICVs operating within its coverage area.

Indeed, the integration of cloud control into ICV systems offers promising benefits, the increased number of communication routes in the cloud control structure may result in undesirable delays in the control loop, which is harmful to the vehicle control performance\cite{chang2019_delay_analysis,cao2021networked}. To tackle this challenge, the design principles of the controller require further investigation. Given that the dynamics of ICV control under regular working conditions can often be approximated by a linear model\cite{chang2019_delay_analysis,pan2023vehicle}, the problem at hand can be abstracted as a delayed linear feedback control problem, which is a classical problem in the filed of Network Control System (NCS).

It is generally reasonable to divide all those delay-handling control strategy into two categories: active and passive methods. For the passive methods, as it name suggests, is to analyzing the system stability containing a certain length of delay and designing robust enough feedback controllers to resist the delay.
State augmentation\cite{zhang2005new,zhang2008new,shi2009output} method and Lyapunov-functional method\cite{li2011new,kwon2013stability,wenTimeDelay14} are two typical and well-studied methods to analyze the stability of the linear control system under certain delay. Since the core idea of the passive method is to search a feasible solution within the space which is compressed by the delay, they might not guarantee workable performance under longer delay cases.
The active methods, in contrast, design more active structure to compensate the impact brought by the delay. Prediction-based control is one of the intuitive method to actively handle the system with longer delay, largely due to the core idea of the method is to adopt the predicted future system states to build the control strategy to compensate for the unstable trend brought by the delay\cite{deng2022predictor}. In 1950s, ``smith predictor"  was firstly proposed in the industry\cite{smith_pre}, to stabilize system contains input dead time. The basic logic of the smith predictor is to introduce a inner system to predict the future system states. Various studies are inspired by the same idea to design a prediction-based controller for the linear feedback system\cite{liu2007design,zhang2012design,gonzalez2012robustness}. Among those prediction methods, the predictor-observer structured controller is a widely discussed prediction structure. As it perfectly adopts the structure of the conventional Luenberger observer and uses linear prediction result to improve the observer's performance. \cite{hao2019output} considered the problem with constant input delay, and proposes sequential predictor-observer to compensate the delay. \cite{gonzalez2019robust} studies the problem with time-varing input delay and proposes delay-free inter-connected system construction to analyze the predictor-observer system's stability. Under the same framework, \cite{gonzalez2019gain} discussed the situation with bi-directional but fully measurable delay. It could be seen that there have been various theory on handling the delay with the prediction-based controller, however, most of them focusing on the theoretically stability analyzation and neglecting the feasibility of certain assumption in the real control plant. Some experimental-oriented studies have realize the prediction-based delay compensation controller in the real plant, however, most of those experimental plants are with simple dynamic, like AC servo motor\cite{lai2009design}, electronic throttle\cite{zhang2018new}, ball and beam system\cite{zhang2012design}, etc. The implementation method and performance of the prediction strategy under larger scale system like a connected vehicle demands to be further studied and evaluated.

This article centers on the challenge of bi-directional delay in cloud control of Intelligent and Connected Vehicles (ICV). A pioneering predictor-observer structured controller is introduced, requiring solely output delay measurements. The key contributions are as follows:

\begin{enumerate}[1)]
	\item A pioneering predictor-observer structured controller is developed to perform delay compensation using solely single directional delay measurements in the connected vehicle cloud control system. The innovative construction of a delay-free interconnected system is introduced to assess stability under specified delay boundaries and model uncertainty, along with corresponding stability criteria.
	\item  The proposed predictor-observer structured controller is rigorously validated via simulations utilizing a high-fidelity vehicle dynamic model. Results confirm its prediction accuracy in the presence of time-varying delay. Moreover, the controller is uniquely applied to a real-world connected vehicle lateral control scenario, highlighting superior performance and enhanced delay resistance compared to conventional state feedback control.
\end{enumerate}
\section{ICV cloud control system modeling}
In the context outlined in the preceding section, the centralized cloud control system operates ideally by gathering vehicle states and surrounding environmental data through a distributed sensor network. This information is then processed and transmitted to the cloud server's decision and control function model. Subsequently, lateral and longitudinal control commands are generated and wirelessly conveyed to the controlled connected vehicle, where they are executed by specific on-board actuators. However, within this intricate process, delays can manifest during networked sensing and control command distribution, as depicted in Fig.~\ref{fig:sch_of_ICV_cloud_control}(a).
From a feedback control perspective, the delay can be classified as output delay for the sensor-controller connection and as input delay for the controller-actuator connection. Notably, the proposed strategy can be extended to a more practical scenario where the connected vehicle solely relies on cloud-sourced sensing information and generates control inputs using an on-board controller, as illustrated in Fig.~\ref{fig:sch_of_ICV_cloud_control}(b). The key distinction lies in considering the physical actuation delay as the pertinent control delay.
Given the relatively weak coupling between lateral and longitudinal dynamics under normal operating conditions, lateral and longitudinal control dynamic is considered separately and could actually cover various scenarios\cite{chang2019_delay_analysis}. It should be stressed in particular that only the lateral control scenario is discussed and tested in this article. The linearized longitudinal controller is designed with the same logic and thus omitted to save space, and the model details could be found in \cite{pan2023cloud,pan2023vehicle,xu2022reachability}.

\begin{figure}[htb]\centering
	\includegraphics[width=1\linewidth]{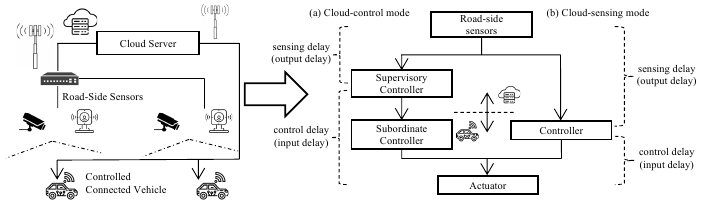}
	\caption{Schematic of the connected vehicle cloud control system under (a) cloud control mode (b) cloud sensing mode}\label{fig:sch_of_ICV_cloud_control}
\end{figure}

\subsection{Vehicle Control Dynamics}
The lateral control task could be interpreted as controlling the vehicle steering to track a predefined trajectory. According to existing research\cite{ackermann1995linear,enache2009driver}, neglecting the external disturbance brought by the wind and other factors, the lateral control dynamic could be linearized as follows after being discretized with a zero-order holder (ZOH):
\begin{align}
	 x(k+1) &=A_{lat} x(k)+ B_{lat} \delta (k) + P_{r} \rho (k), \label{eq:lat_sys_dynamic}\\
	 \delta (k)&=K_{lat}x(k), \label{eq:lateral_ip_def} \\
	 x(k)&=\left[ \beta (k),r(k), \psi _{L}(k),y_{L}(k) \right] ^\mathrm{T}, \nonumber
\end{align}
where $A_{lat}$ is the state matrix, $B_{lat}$ is the input matrix, $P_r$ is the feed-forward input matrix, $K_{lat}$ is the feedback control gain, $x$ is the state vector, and it is composed of: the vehicle slide angle $\beta$ , the yaw rate $r$, the relative yaw angle $\psi_L$, and the lateral offset at preview distance $y_L$, $\delta$ is the expected steering angle of the front wheel, $\rho$ is the curvature of the trajectory. \iffalse The physical implication of aforementioned system states are shown in Fig.\ref{fig:lateral_dyn_imp}. \fi The discrete system matrices $A_{lat}, B_{lat}$ and $P_r$ are determined by their continuous forms, which are given by:
\begin{align}
    A_{lat_c}=&\left[\begin{matrix}
		- \frac{c_f+c_r}{mv} & -1+ \frac{c_rl_r-c_fl_f}{mv^2} & 0 & 0 \\
	\frac{c_rl_r-c_fl_f}{I_z} & - \frac{c_fl_{f}^{2}+c_rl_r^2}{vI_z} & 0 & 0 \\
	0 & 1 & 0 & 0 \\
	v & l_s & v& 0
\end{matrix}\right],  \label{eq:lat_mtx_def_1}  \\
B_{lat_c}=&\left[\frac{c_f}{mv},\frac{c_{f}l_{f}}{I_{z}},0,0\right]^T ,  \label{eq:lat_mtx_def_2}  \\
P_{r_c}=&\left[0,0,-v,-vl_{s}\right]^T, \label{eq:lat_mtx_def_3}
\end{align}
where $c_f$ and $c_r$ are the equivalent cornering stiffness of the front and rear axis, respectively. $l_f$ and $l_r$ are the distance from $CG$ (Center of Gravity) to front and rear axle, $ m$ is the vehicle total mass, $I_z$ is the vehicle yaw moment of inertia, $v$ is the longitudinal driving speed, and $l_s$ is the preview distance.

\subsection{Control Problem Statement}
As both the lateral and longitudinal dynamics could be simplified as a linear feedback control system, the problem could be abstracted as a linear output feedback control problem under time-varying input and output delay. The linearized system dynamic with bounded delay and model uncertainty could be indicated as follow:
\begin{align}
    x(k+1)&=(A+\Delta A(k)) x(k)+(B+\Delta B(k))u(k-d_k^I), \label{eq:abs_sys_1} \\
    y(k)&=Cx(k-d_k^O), \label{eq:abs_sys_2} \\
    u\left(k\right)&=f(y(k)), \label{eq:abs_sys_3}
\end{align}
where $f(y(k))$ is the feedback strategy to be determined, $d_k^I\in\left[h_1^I,h_2^I\right]$ is the bounded time-varying input delay, $d_k^O\in\left[h_1^O,h_2^O\right]$ is the bounded time-varying output delay. $\Delta A(k)$ and $\Delta B(k)$ are time-varying model uncertainties, and are also assumed to be bounded as \cite{han2001robust}:
\begin{equation}
    (\Delta A(k),\Delta B(k))= \gamma E \Delta(k)(H_A,H_B),
\end{equation}
where $\Delta(k)$ is a time-varying matrix satisfying $\Delta(k)^T\Delta(k)\leq I$,\ $E,H_A,H_B$ are time-constant matrices to determine the size of the model uncertainty. $\gamma$ is also a positive scalar that determine the size of the uncertainty, and could be used for further optimization. It is also assumed that the output feedback system is controllable and observable.
\section{Controller design and stability analysis}
As (\ref{eq:abs_sys_1}) and (\ref{eq:abs_sys_2}) already provides the delayed linear feedback system modeling for the concerned system, the problem comes as designing the feedback strategy $u(k)=f(y(k))$. Therefore, in this section, the predictor-observer structured delay compensation controller is introduced, and the stability analysis for the close-loop system is also conducted.
\subsection{Predictor-observer structured Controller}
Corresponding to (\ref{eq:abs_sys_1}) and (\ref{eq:abs_sys_2}), designing a controller with predictor-observer structure as (\ref{eq:ctl_u_eq})-(\ref{eq:z_pre_dyn}), which only demands measurement on output delay.
\begin{align}
    u(k)&=K{\hat{Z}}(k) \label{eq:ctl_u_eq},\\
    {\hat{Z}}(k+1)&=A{\hat{Z}}(k)+Fu(k)+L\left(A^{d_k^O}{\bar{y}}(k)-C{\hat{Z}}(k)\right), \label{eq:z_pre_dyn}
\end{align}
where:
\begin{align}
F=&(A^{{-h}_1^I}+A^{-h_2^I})\frac{B}{2},  \\
    {\bar{y}}(k)=&y(k)+CA^{{-d}_k^O}\left(\Phi_k\left(h_1^I\right)+\Phi_k\left(h_2^I\right)+{\bar{\mathrm{\Omega}}}_k\left(d_k^O\right)\right) \label{eq:y_est}, \\
    {\bar{\mathrm{\Omega}}}_k\left(d_k^O\right)=&\frac{1}{2}\sum_{i=0}^{d_k^O-1}A^{d_k^O-i-1}Bu(k-d_k^O+i-h_1^I) +\frac{1}{2}\sum_{i=0}^{d_k^O-1}A^{d_k^O-i-1}Bu(k-d_k^O+i-h_2^I), \label{eq:omega_est} \\
        \Phi_k\left(h_f^I\right)=&\frac{1}{2}\sum_{i=0}^{h_f^I-1}{A^{-i-1}Bu(k-h_f^I+i)},\ f=1,2.
\end{align}
where $K $ and $L$ are the controller and observer gains to be determined. It is necessary to define the equivalent state variable $Z(K)$, which subject to Artstein's state transformation as:
\begin{equation}
Z(K)=x(k)+\Phi_k\left(h_1^I\right)+\Phi_k\left(h_2^I\right), \label{eq:z_def}
\end{equation}
and $\hat{Z}(k)$ is an estimation of $Z(k)$ inside the observer.

The prediction logic of the proposed method could be indicated as Fig.~\ref{fig:prd_logic}.
From the timing relation indicated through the Figure and the prediction function, (\ref{eq:z_pre_dyn}), (\ref{eq:y_est}) and (\ref{eq:omega_est}) indicate that the predictor-observer is trying to conduct prediction of states $Z(k)$, based on the state information coming from $d_k^O$ steps before. It should also be conscious of the fact that, $Z(k)$ is a reflection on future system states (the future moment when the control signal finally arrives at the actuator). Therefore, (\ref{eq:ctl_u_eq}) indicates that the controller determines the control input based on future system states. This could also clarify how the proposed strategy compensates for the input delay actively. (\ref{eq:omega_est}) indicates that the prediction corresponding to compensate input delay and output delay is coupled. Thus, as the input delay is assumed to be unmeasurable, the estimation of the input delay will influence the prediction corresponding to the output delay, and the model uncertainty will worsen this issue. Thus, in the next sub-section, stability will analyzed considering this issue.

\begin{figure}[htb]
    \centering
    \includegraphics[width=0.8\linewidth]{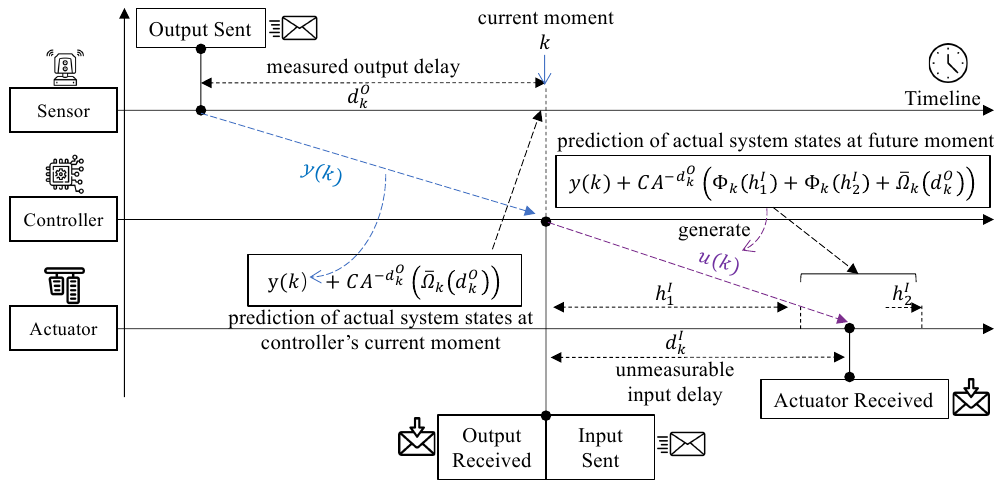}
    \caption{Prediction logic of the proposed method.}
    \label{fig:prd_logic}
\end{figure}
\begin{remark}
    According to the timing relation from the Fig.~\ref{fig:prd_logic}, it could be aware that, the output delay $d_k^O$ is technically observable at the moment when the controller determine the control input, whereas the input delay doesn't. The input delay is only possible to be measured when the control input arrives at the actuator, and will introduce extra process on achieving the measurement and corresponding control input selection. Thus, the assumption in this article is more rational and useful under real control scenario, comparing similar study\cite{gonzalez2019gain} demands full delay measurement.
\end{remark}
\begin{remark}
    There is an update on the math operation of the predictor-observer. The original one seems not robust with long delay range, and the $A^h$ operation may result in infinite integration. Thus, based on the both simulation and real test experience, a new predictor-observer function is proposed, which is pretty stable to the largely varying delay.
\end{remark}
\subsection{Inter-connected system modeling and stability analysis}
To analyze the input-output stability of the predictor-observer structured feedback control system within specific delay boundaries and under model uncertainty, the conventional procedure involves constructing an inter-connected system. As expounded in the preceding section, the intricate interplay of input and output delay compensation adds complexity, rendering the construction of a delay-free interconnected system less straightforward than in previous studies.  
 Therefore, the main theorem is introduced to establish a description of a delay-free interconnected system for the closed-loop system defined by (\ref{eq:abs_sys_1})-(\ref{eq:abs_sys_2}), coupled with the prediction-based control strategy outlined in (\ref{eq:ctl_u_eq})-(\ref{eq:y_est}).

\begin{theorem}\label{thm:int_cont_sys_rep}
    The closed-loop system formed by (\ref{eq:abs_sys_1})-(\ref{eq:abs_sys_2}) and prediction-based control law (\ref{eq:ctl_u_eq})-(\ref{eq:y_est}) could be modeled equally as the following interconnected system:
    \begin{align}
M_s&: \begin{cases}
	\bar{X}(k+1)&=\bar{A}\bar{X}(k)+\bar{B}\bar{W}(k), \\
	\bar{Y}(k)&=\bar{C}\bar{X}(k)+\bar{D}\bar{W}(k),
\end{cases} \label{int_con_sys_for_1} \\
\Delta &: \bar{W}(k)=\bar{\Delta}(k)\bar{Y}(k) , \label{int_con_sys_rev_1}
\end{align}
where $\bar{\Delta}(k)$ is an uncertain system satisfying $\bar{\Delta}(k)^T\bar{\Delta}(k)\leq I, \forall k \geq 0$. The augmented state/output/input are defined as:
\begin{align*}
    \bar{X}(k)&=\left[{Z(k)}\ {x(k-1)}\ u(k-1)\ e(k)\right]^T\\
    \bar{Y}(k)&=\left[
    \begin{smallmatrix}
y_{\Delta}(k)& v(k)&\widetilde{v}(k)&y_p(k)&v(k)&v(k)&v(k)&v(k)&v(k)&v(k)&q(k)
    \end{smallmatrix}\right]^T  \\
    \bar{W}(k) &=\left[
    \begin{smallmatrix}
        w_{\Delta}(k)&w_d(k)&{\widetilde{w}}_{d}(k)&w_p(k)&w_1(k) &w_2(k)&w_3(k)&w_4(k)&w_5(k)&w_6(k)& w_7(k)
    \end{smallmatrix}\right]^T
\end{align*}
where:
\begin{align*}
    v(k)=&u(k)-u(k-1),\ q(k)=x(k)-x(k-1) \\
         {\widetilde{v}}(k)=&\sum_{i=0}^{d_k^O-1}{A^{d_k^O-i-1}B}v(k), \\
         {\widetilde{w}}_{d}(k)=&\sum_{i=0}^{d_k^O-1}{A^{d_k^O-i-1}B}w_{d}(k-d_k^O+i)
\end{align*}
and:
\begin{align*}
    \bar{A} &=\left[\begin{smallmatrix}
        A+FK & 0& 0& -FK\\I&0&-\Upsilon &0\\ K & 0&0&-K \\ 0&0&0&A-LA^{d_k^O}CA^{{-d}_k^O}
    \end{smallmatrix}\right],   \\
    \bar{B} &=\left[\begin{smallmatrix}
        \gamma E & \frac{\tau}{2}B & 0&0&0&0&0&0&0&0&0 \\
        0&0&0&0&0&\mu_2 &0&0&0&0&0 \\
        0&0&0&0&0&0&0&0&0&0&0 \\
        \gamma E & \frac{\tau}{2}B &0 & LA^{d_k^O}CA^{{-d}_k^O}\widetilde{\gamma}\widetilde{E} &0&0& 0&0&0&0&0\end{smallmatrix}\right], \\
       \bar{C} &=\left[\begin{smallmatrix}
           H_A&0&H_B-H_A\Upsilon&0 \\ K & 0&-I & -K \\
          \beta_2B K & 0&-\beta_2B & -\beta_2BK \\0 & \beta_1H_A & \beta_2H_B+\beta_3H_AB &0\\
          K & 0 &-I &-K \\ K & 0 &-I &-K  \\ K & 0 &-I &-K \\ K & 0 &-I &-K \\ K & 0 &-I &-K \\K & 0 &-I &-K \\I&-I&-\Upsilon &0
       \end{smallmatrix}\right],  \\
       \bar{D} &=    \left[\begin{smallmatrix}
        0&\frac{\tau H_B}{2} &0&0&-\mu_1H_B & \mu_2H_A &0&0&0&0&0\\
        0&0&0&0&0&0&0&0&0&0&0\\
        0&0&0&0&0&0&0&0&0&0&0\\
        0&0&I&0&0&0&-\frac{1}{2}\mu_3&-\frac{1}{2}C_{d_k^O}^1\mu_4&-\frac{1}{2}\mu_5&-\frac{1}{2}C_{d_k^O}^1\mu_6&-\mu_7\beta_1H_A\\
        0&0&0&0&0&0&0&0&0&0&0\\
        0&0&0&0&0&0&0&0&0&0&0\\
        0&0&0&0&0&0&0&0&0&0&0\\
        0&0&0&0&0&0&0&0&0&0&0\\
        0&0&0&0&0&0&0&0&0&0&0\\
        0&0&0&0&0&0&0&0&0&0&0\\
        0&0&0&0&0&\mu_2&0&0&0&0&0\\
    \end{smallmatrix}\right].
\end{align*}
and:
\begin{align*}
    \mu_1=& \Vert \frac{1}{2}\sum_{j=1}^{2}\sum_{f=1}^{h_j^I-1}z^{-f} \Vert _{\infty}, \ \mu_2=\Vert \frac{1}{2}\sum_{j=1}^{2}\sum_{i=0}^{h_j^I}\sum_{f=1}^{h_j^I-i-1}{A^{-i-1}Bz^{-f}}  \Vert _{\infty}, \\
    \mu_3=&\Vert \sum_{i=0}^{d_k^O-1}\sum_{f=1}^{d_k^O+h_1^I-i-1}{A^{d_k^O-i-1}H_Bz^{-f}} \Vert _{\infty},\ \mu_4=\Vert \sum_{i=0}^{d_k^O-1}\sum_{f=1}^{d_k^O+h_1^I-i-1}{A^{d_k^O-i-2}H_A B z^{-f}} \Vert _{\infty},  \\
    \mu_5=&\Vert \sum_{i=0}^{d_k^O-1}\sum_{f=1}^{d_k^O+h_2^I-i-1}{A^{d_k^O-i-1}H_Bz^{-f}} \Vert _{\infty},\ \mu_6 = \Vert \sum_{i=0}^{d_k^O-1}\sum_{f=1}^{d_k^O+h_2^I-i-1}{A^{d_k^O-i-2}H_A B z^{-f}} \Vert _{\infty}, \\
   \mu_7 =& \Vert \sum_{f=0}^{d_k^O-1}z^{-f} \Vert _{\infty},\\
   \beta_1=& C_{d_k^O}^1A^{d_k^O-1} , \beta_2=\sum_{i=0}^{d_k^O-1}A^{d_k^O-i-1}, \ \beta_3=\sum_{i=0}^{d_k^O-1}\left(C_{d_k^O}^1A^{d_k^O-i-2}\right),\\
   \Upsilon=&\frac{1}{2}\sum_{j=1}^{2}\sum_{i=0}^{h_j^I}{A^{-i-1}B},
\end{align*}
\end{theorem}

\begin{proof}
    This is an extension deduction from result in \cite{pan2023cloud}, and adopt the math tools from \cite{gonzalez2019robust} and \cite{gonzalez2013robust}. See \ref{apx:1} for details.
\end{proof}

Since the predictor-observer structured feedback control system dynamic could be simplified as the interconnected system representation from Theorem~\ref{thm:int_cont_sys_rep}, various routine tools could adopted to conduct the stability condition. A optional lemma to conduct the stability analysis of the inter-connected system is provided as follow.
\begin{lemma}\label{lm_stab_cond}
\cite{pan2023cloud}Considering an interconnected system with structure as:
\begin{align}
M_s&: \begin{cases}
	\bar{X}(k+1)&=\bar{A}\bar{X}(k)+\bar{B}\bar{W}(k), \\
	\bar{Y}(k)&=\bar{C}\bar{X}(k)+\bar{D}{\bar{W}}(k),
\end{cases} \label{int_con_sys_for_1_lm} \\
\Delta &: \bar{W}(k)=\bar{\Delta}(k)\bar{Y}(k).  \label{int_con_sys_rev_1_lm}
\end{align}
The system is robustly asymptotically stable with decay rate $\beta$ if there exist symmetric matrix $P>0$ such that the following matrix inequality is satisfied:
\begin{equation}
    \left[\begin{matrix}-\beta^2P&0&{\bar{A}}^TP&{\bar{C}}^T\\\ast&-I&{\bar{B}}^TP&{\bar{D}}^T\\\ast&\ast&-P&0\\\ast&\ast&\ast&-I\\\end{matrix}\right]<0.  \label{MI_cond}
\end{equation}
\end{lemma}

The viability of stability analysis in the format of matrix inequality, as expressed in (\ref{MI_cond}), enables the tailored design of stabilizing controller and observer gains. It's important to underscore that if the controller and observer gains are unknown variables, the matrix inequality (\ref{MI_cond}) entails inverse relation constraints and is inherently unsolvable through direct means. Nevertheless, the Cone Complementary Linearization (CCL) methodology, a well-established tool, has proven effective in addressing analogous challenges \cite{gonzalez2013robust,pan2023vehicle}. Its adoption offers a pathway to deducing feasible controller gains. However, for brevity, the specific controller design algorithm is omitted from this discussion.

\begin{remark}
    The theoretical contributions of this article in comparison to related studies warrant specific emphasis. The foundational concept behind the proposed predictor-observer structured controller finds its roots in \cite{gonzalez2019gain}. However, this previous work falls short in addressing scenarios involving unmeasurable input delay and model uncertainty. The controller's structural aspects have been discussed in our recently published study \cite{pan2023cloud}, yet without accounting for the issue of model uncertainty-a gap that persists even in the absence of high-fidelity simulations and real vehicle testing. Furthermore, it's noteworthy to acknowledge that elements of the mathematical framework used to model impact of the delay draw inspiration from \cite{gonzalez2013robust}.
\end{remark}
\section{Simulation test}
For deploying the control strategy to an actual system, preliminary simulation is essential. Therefore, this section establishes a simulation scenario using Matlab Simulink + Carsim, integrating a high-fidelity vehicle dynamic model. The proposed control algorithm's validation occurs within lateral control scenarios, enabling comprehensive evaluation of controller performance.
\subsection{Simulation setting}
In the lateral control simulation, an A-class hatchback vehicle model from Carsim is adopted to simulate the vehicle's lateral dynamic. To make the simulation result closer to the real plant situation, some random white noise is added to the Carsim outputs, and the amplitude of specific noise is determined by certain sensor accuracy. The structure of the lateral control simulation could be shown as Fig. \ref{fig:lat_sim_struc}. The key parameters of the model are shown in Table~\ref{tab:lateral_para}.

\begin{figure}[htb]
    \centering
    \includegraphics[width=0.9\linewidth]{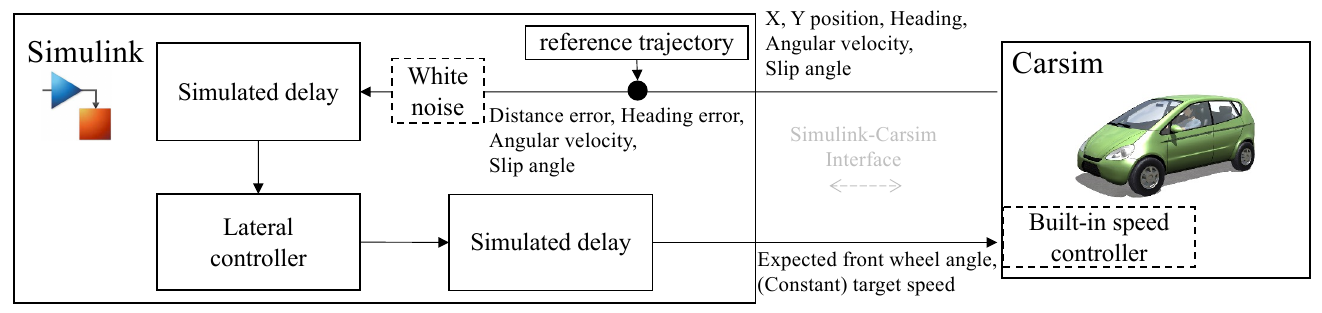}
    \caption{The structure of the lateral control simulation}
    \label{fig:lat_sim_struc}
\end{figure}

\begin{table}[ht]
	\centering
	\caption{Key parameters of lateral control dynamic model in the simulation\label{tab:lateral_para}}
	\begin{tabular} {cccccccc}
	\toprule
$l_f/m$ & $l_r/m$ &$l_s/m$ & $m/kg$  & $I_z/(kg\cdot m^2)$ & $v/(m/s)$   &$c_f/(N/rad)$ & $c_r/(N/rad)$ \\
\midrule
1.1 &	1.25&2.5 &	850.8 &	750 & 5	&	35696&	32299   \\
\bottomrule
\end{tabular}	
\end{table}
The vehicle is assumed to be controlled by a remote controller to tack a reference lane-change trajectory, shown in Fig.~\ref{fig:lat_ref_traj_sim}. The time-varying simulated delay is added at bilateral side of the controller. It needs to be noted that the output and input delay considered in the simulation are discretized by the control cycle, shown as:
\begin{equation}
    T_k^O=d_k^O\cdot T_c,\ T_k^I= d_k^I\cdot T_c, \label{eq:sim_dl_def}
\end{equation}
where $T_k^O$ is the output delay at step $k$, $T_k^I$ is the input delay at step $k$, $T_c$ is the control cycle, which is set to be 0.05$s$ in the simulation. The delay settings (step delay length boundary) considered in the lateral simulation are $d_k^O\in\left[4,7\right],\ d_k^I \in \left[3,5\right]$. It also needs to be clarified that both the input and the output delay are generated randomly with uniform distribution within the boundary.

\begin{figure}[htb]
    \centering
    \includegraphics[width=0.9\linewidth]{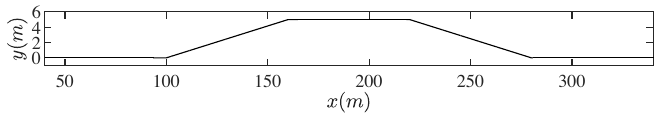}
    \caption{Reference trajectory in the lateral simulation.}
    \label{fig:lat_ref_traj_sim}
\end{figure}

A discrete LQR controller is set as a comparison controller. Based on the same linear model, the controller gain are given in Table.~\ref{tab:lateral_controller_gain_sim}. An MPC controller is also designed with the Simulink MPC toolbox, whose major design parameters are indicated as Table.~\ref{tab:lat_mpc_par_sim}. In comparison, the proposed predictor-observer structured controller is designed, and have been verified to satisfy the stability condition proposed in previous section. The controller and observer gains are given also in Table.~\ref{tab:lateral_controller_gain_sim}. A prediction controller (from \cite{gonzalez2019gain})  with similar structure but demand input delay measurement is also adopted as a comparison controller. To make the comparison more intuitive, the controller is with the same controller and observer gains (the only difference is the prediction structure).

\begin{table}[ht]
	\centering
	\caption{Controller and observer gains for lateral control in the simulation\label{tab:lateral_controller_gain_sim}}
	\setlength\tabcolsep{2pt}
	\begin{tabular}{ccc}
\toprule
Type & Controller gain & Observer gain  \\
\midrule
LQR &$\left[\-0.0309,-0.0210,-0.5149,-0.1810\right]$  &N/A \\
proposed method~/ \cite{gonzalez2019robust}
&$\left[-0.0303,-0.0221,-0.696,-0.1810\right]$&$\left[\begin{smallmatrix}-0.5483&-0.006&0&0\\0.0197&-0.6681&0&0\\0.0011&0.0184&0.25&0\\0.1275&0.0474&0.25&0.25\\\end{smallmatrix}\right]$ \\[2.2ex]
		\midrule
\end{tabular}	
\end{table}
\begin{table}[ht]
	\centering
	\caption{Parameters of the comparison MPC controller in the simulation\label{tab:lat_mpc_par_sim}}
	\begin{tabular}{ccccc}
\toprule
Output weight & Input weight  & Prediction horizon&Control horizon& Input bound  \\
\midrule
 $\left[0.05,0.1,0.1,0.2\right]$  &	1 &15&3&$\pm$ 0.2\\
\bottomrule
\end{tabular}	
\end{table}

\begin{remark}
Importantly, the Carsim's inherent math model presents nonlinearity\cite{liu2022recurrent}. The linearized model parameters tabulated in Table \ref{tab:lateral_para}, such as tire cornering stiffness, signify estimated values that have been meticulously verified for their state prediction accuracy. Consequently, the utilization of the proposed method in the simulation, encompassing the Carsim model, underscores its noteworthy robustness in the face of model uncertainty. This remark holds true for the subsequent real vehicle testing section as well.
\end{remark}

\subsection{Simulation result and analysis}
The controlled vehicle trajectories in the lateral control simulation are indicated as Fig.~\ref{fig:lat_traj_sim}. The controlled trajectory of the regular controller under near zero delay condition (one-step input and output delay is retained, to remain the rationality of the simulation) are also given as a comparison. The resultant trajectory clearly reveal that  conventional LQR and MPC controllers excel in trajectory tracking with minor delays. However, their performance degrades markedly in the presence of significant input and output delays, even leading to unstable tracking. This observation underscores the limitations of applying standard lateral controllers to cloud control scenarios.
In contrast, the proposed method could alleviate the situation, and comparative result of the two prediction-based controllers reveals a similar tracking efficacy, whereas, the proposed method achieves this level of performance without necessitating precise measurement of input delay duration.

To better illustrate how the prediction generated by the controller compensate the bi-directional delay effect, the evolution of the current system states, system states after closed-loop delay, and predicted system states through the simulation tests  presented as Fig.~\ref{fig:syn_st_lat_sim}. It need to be explained that Fig.~\ref{fig:syn_st_lat_sim} is generated in the timeline of the actuator, and the states $2(A^{-h^I_1}+A^{-h^I_2})^{-1}\hat{Z}(k)$ is adopted as the prediction of the system sates when the control signal is finally arrived at the actuator, it is also assumed that a message contain the system output are sent from the sensor and forward by the controller and finally arrive at the actuator(along with the control input) and is adopt as the delayed system states in Fig.~\ref{fig:syn_st_lat_sim}. The results demonstrate rapid convergence of predicted system states to actual states post-impulse disturbance, providing insight into the notable improvement in control performance.
\begin{figure}[htb]
    \centering
    \includegraphics[width=0.7\linewidth]{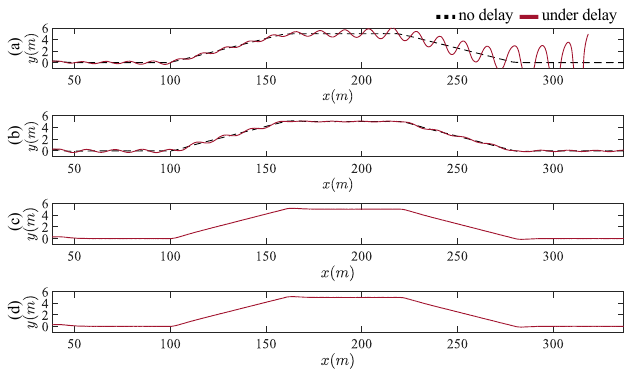}
    \caption{Controlled vehicle trajectory under (a) LQR controller (b) MPC controller   (c) \cite{gonzalez2019gain} controller (d) proposed controller}
    \label{fig:lat_traj_sim} 
\end{figure}

\begin{figure}[htb]
    \centering
    \includegraphics[width=0.9\linewidth]{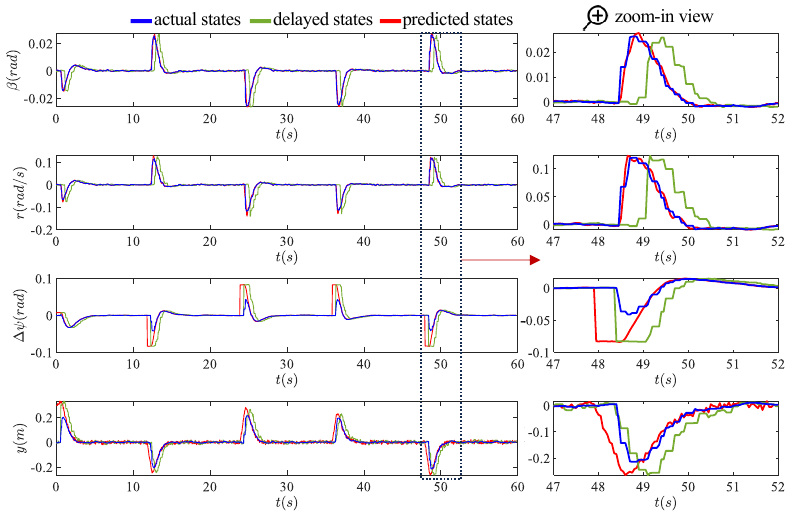}
    \caption{Evolution of the current system states, delayed (close-loop delay) system states and predicted system states}
    \label{fig:syn_st_lat_sim}
\end{figure}

\section{Real vehicle test}
In this section, the proposed control strategy is deployed on a real ICV, and the lateral control test is conducted. The ensuing subsection presents both the test plant configuration and the corresponding test results.

\subsection{Test setting}
The foundation of the test plant is a 2018 ChangAn CS55 SUV, shown as Fig.~\ref{fig:CS55_ext_view}, is equipped with relevant autonomous driving sensors and devices. Notably, for the test purposes, only the IMU (Inertial Measurement Unit) is utilized to ensure accurate position information of the controlled vehicle. The controller logic is achieved through the Matlab on a personal computer, and communication links are built between the Matlab and original sensor and actuator drivers that running on the on-board IPC (Industrial Personal Computer). This interaction is established via TCP protocol, physically facilitated by a direct Ethernet cable connection. The aforementioned test plant structure is indicated as Fig.~\ref{fig:test_plant_sys_struc}. Following the same process as in the lateral simulation, the main lateral control dynamic parameters of the plant are indicated as Table.\ref{tab:lateral_para_CS55}.

\begin{figure}[htb]
    \centering
    \includegraphics[width=0.8\linewidth]{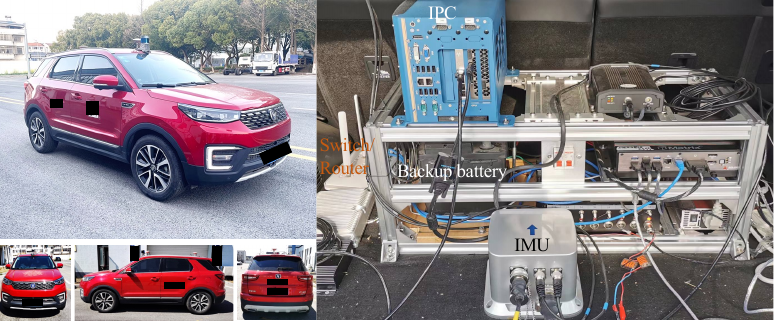}
    \caption{The ChangAn CS55 ICV test plant}
    \label{fig:CS55_ext_view}
\end{figure}

\begin{figure}
    \centering
    \includegraphics[width=0.9\linewidth]{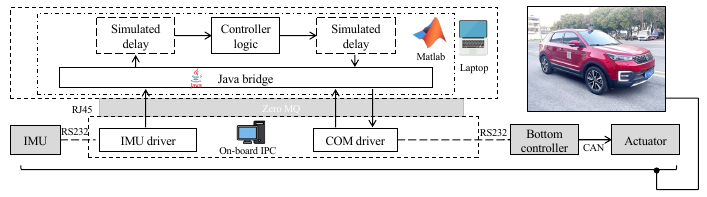}
    \caption{Structure of the test system}
    \label{fig:test_plant_sys_struc}
\end{figure}

\begin{table}[ht]
	\centering
	\caption{Key parameters of lateral control dynamic model in the
 field test\label{tab:lateral_para_CS55}}

	\begin{tabular} {cccccccc}
        \toprule
$l_f/m$ & $l_r/m$ &$l_s/m$ & $m/kg$  & $I_z/(kg\cdot m^2)$ & $v/(m/s)$   &$c_f/(N/rad)$ & $c_r/(N/rad)$ \\
\midrule
1.17 &	1.48&3 &	1570 &	2700 & 3.5	&	121100&	199831   \\
\bottomrule
\end{tabular}	
\end{table}

Same as the scenario setting in the simulation part, the vehicle is also controlled to track a reference lane change trajectory within a twin four-lane road, shown as Fig.~\ref{fig:test_field_map_and_ref_traj}. Regarding the delay configuration, as illustrated in Fig. \ref{fig:test_plant_sys_struc}, an extra simulated delay ($\widetilde{d}^O_k$ for output delay and $\widetilde{d}^I_k$ for input delay) is introduced within Matlab for the sake of convenient delay setting. It should be noted that the original system  contain certain length of inherent delay ($\bar{d}^O_k$ for output delay and $\bar{d}^I_k$ for input delay), due to the communication structure and the steering actuator dead time, which are not neglectable in the proposed controller design. After carefully measure and estimate the system inherent delay, the inherent delay range, additional simulated delay range, the aggregate delay range and the discrete delay step-size $T_c$ (also the control step-size) are indicated as Table.\ref{tab:vec_test_delay_set}.
\begin{figure}[htb]
    \centering
    \includegraphics[width=0.8\linewidth]{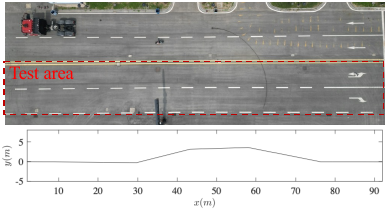}
    \caption{Test site overhead view and the reference trajectory}
    \label{fig:test_field_map_and_ref_traj}
\end{figure}
\begin{table}[ht]
	\centering
	\caption{Delay range settings in field test\label{tab:vec_test_delay_set}}
	\begin{tabular}{ccccccc}
		\toprule
		$\bar{d}^O_k$ & $\bar{d}^I_k$ &$\widetilde{d}^O_k$ & $\widetilde{d}^I_k$ & $d^O_k$ &$d^I_k$&$T_c(s)$ \\
		\midrule
		$[1,1]$ & $[2,6]$ & $[5,6]$ & $[2,2]$  & $[6,7]$ & $[4,8]$&0.05\\
        \bottomrule
	\end{tabular}
\end{table}

The LQR controller is also adopted as the comparison controller, and a set of controller gain that tuned under scenarios without additional delay are given in Table.~\ref{tab:lateral_controller_gain_vec_test}.
It is crucial to emphasize that in the field test, the slipping angle is omitted, causing the lateral control model to simplify to a 3-dimensional model. This simplification arises from the relatively minor slipping angle observed in the testing scenario. Additionally, it is challenging to obtain a direct and precise measurement of the slipping angle using the existing on-board IMU. 
 In contrast, the proposed predictor-observer structured controller is developed based on the delay configuration detailed in Table ~\ref{tab:vec_test_delay_set}. These design choices in line with the stability conditions outlined in this article. The corresponding controller and observer gains are provided in Table ~\ref{tab:lateral_controller_gain_vec_test}.
To mitigate potential unidentified randomness within the test plant, a total of 10 tracking tests are conducted. These tests encompass the delay settings, utilizing both the LQR controller and the proposed controller. Additionally, to serve as a basis of comparison, repeated tests are carried out with the LQR controller under conditions without additional delay (i.e., only inherent delay).

\begin{table}[ht]
	\centering
	\caption{Controller and observer gains for lateral control in field test\label{tab:lateral_controller_gain_vec_test}}
	\begin{tabular}{ccc}
		\toprule
Type & Controller gain & Observer gain  \\
\midrule
LQR &$\left[\begin{smallmatrix}-0.0249&-0.7709&-0.2594\\\end{smallmatrix}\right]$  &N/A \\
proposed method&$\left[\begin{smallmatrix}-0.0249&-0.7709&-0.2594\\\end{smallmatrix}\right]$&$\left[\begin{smallmatrix}-0.6545&0&0\\-0.0141&0.3&0\\0.0492&0.1500&0.3000\\\end{smallmatrix}\right]$ \\
	\bottomrule
\end{tabular}	
\end{table}

\subsection{Test result and analysis}
In absence of additional delay, four trajectory outcomes governed by the LQR controller are selected and depicted in Fig.~\ref{fig:vec_traj_LQR_no_delay}. These trajectory results indicate the LQR controller's ability to accurately and swiftly track the reference trajectory, despite variations in the initial position. Importantly, as these outcomes are achieved within inherent input and output delays, they suggest the potential robustness of the LQR controller against delay, given meticulous tuning.
\begin{figure}[htb]
    \centering
    \includegraphics[width=1\linewidth]{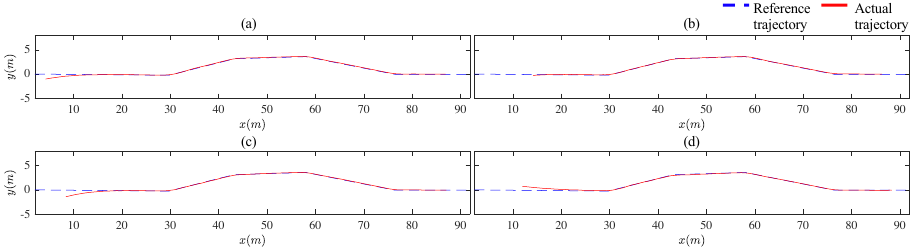}
    \caption{Controlled vehicle trajectory (LQR, no additional delay)}
    \label{fig:vec_traj_LQR_no_delay}
\end{figure}

In the presence of additional delays, four sets of resulting trajectories governed by the LQR controller are chosen and illustrated in Fig.~\ref{fig:vec_traj_LQR_with_delay}. The trajectory outcomes underscore the evident impact of additional delay on lateral control performance. In certain instances, this influence leads to unstable tracking.
\begin{figure}[htb]
    \centering
    \includegraphics[width=1\linewidth]{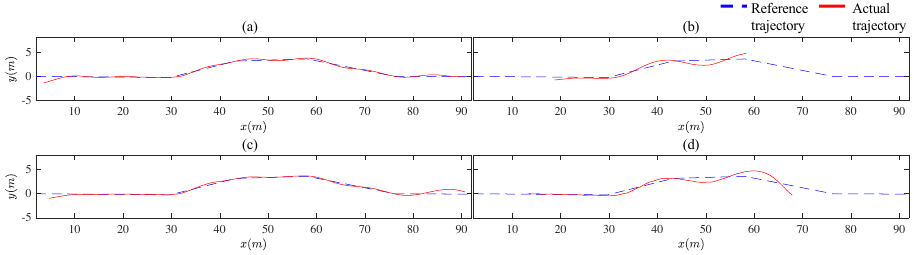}
    \caption{Controlled vehicle trajectory (LQR, with additional delay)}
    \label{fig:vec_traj_LQR_with_delay}
\end{figure}

In contrast, under the same delay configuration, four sets of resulting trajectories governed by the proposed controller are chosen and depicted in Fig.~\ref{fig:vec_traj_syn_with_delay}. The discernible enhancement in control performance is evident when contrasted with the standard LQR controller. Remarkably, the proposed controller consistently achieves accurate tracking of the predetermined trajectory in all additional delay instances. To better convey that the system states are adopted with delay and are perfectly stabilized by the controller, Fig.~\ref{fig:vec_test_syn_st} depicts both the actual system states and the delayed system states that are fed to the controller throughout the test shown in Fig.~\ref{fig:vec_traj_syn_with_delay}.

\begin{figure}[htb]
    \centering
    \includegraphics[width=1\linewidth]{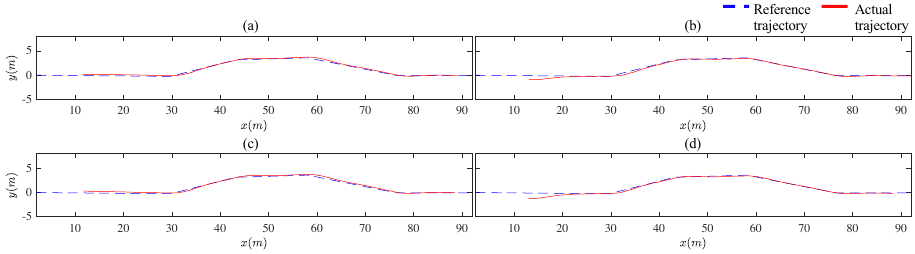}
    \caption{Controlled vehicle trajectory (proposed method, with additional delay)}
    \label{fig:vec_traj_syn_with_delay}
\end{figure}

\begin{figure}[htb]
    \centering
    \includegraphics[width=1\linewidth]{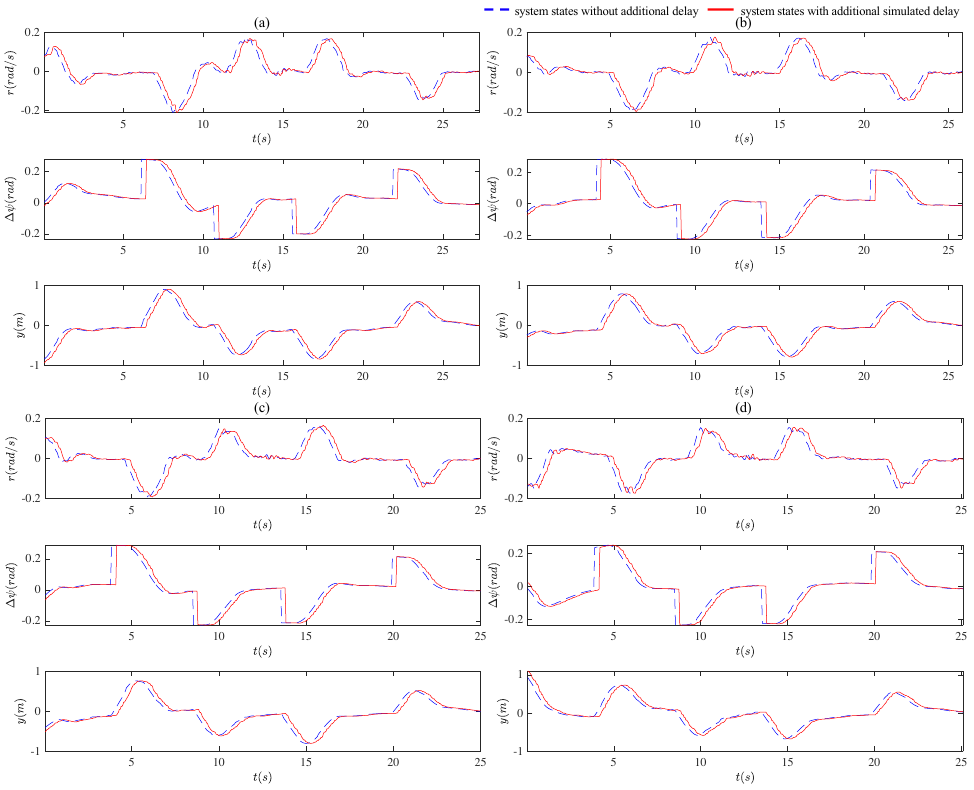}
    \caption{Evolution of the real system states and the delayed system states adopted by the controller}
    \label{fig:vec_test_syn_st}
\end{figure}

Lastly, for comprehensive control performance depiction across all tests and to unveil the repeatability of the proposed approach, box plots are presented as Fig.~\ref{fig:vec_test_boxplot}. These plots indicate the average absolute distance error and heading error relative to the reference trajectory ($y_L$ and $\psi _L$) throughout the 10 repeated tests. Notably, the results unmistakably underscore the impact of delay on the control performance of the LQR controller, while also affirming the consistent effectiveness of the proposed method.
\begin{figure}[htb]
    \centering
    \includegraphics[width=0.8\linewidth]{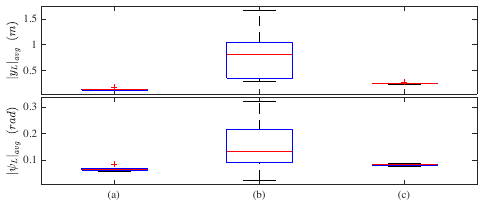}
    \caption{Average absolute distance error and heading error result box plot for (a) LQR without additional delay, (b) LQR with additional delay (c) proposed method with additional delay.}
    \label{fig:vec_test_boxplot}
\end{figure}

\begin{remark}
It need to emphasis that unlike the result in the simulation, no visual prediction results are presented in in the real vehicle test. This is attributed to the engineering challenge of designing an ultimate ``referee"  at the timeline of the actuator to generate replicate results akin to those in Figure.~\ref{fig:syn_st_lat_sim}.  Consequently, only macroscopical result are presented.
\end{remark}

\begin{remark}
  It is essential to emphasize that tests of the proposed controller under conditions of no additional delay were intentionally omitted. This choice is motivated by the fact that the proposed controller is explicitly designed for scenarios involving more significant delays. Its performance under minor delay conditions would closely resemble that of a standard constant-parameter controller, rendering a comparison insignificant. Furthermore, in alignment with the structure of the proposed method, it's important to note that under zero delay conditions, the proposed method would effectively revert to a conventional linear state feedback controller.
\end{remark}

\section{Conclusion}
In summary, this paper presents a pioneering predictor-observer structured delay compensation approach tailored for the ICV cloud control system, addressing the challenges of bi-directional time-varying delay. The introduced inter-connected system modeling technique enables the analysis of input-output stability within the predictor-observer framework, even amidst bounded delays and model uncertainties.
Our proposed control strategy is validated through simulations of connected vehicle lateral control scenarios, employing a high-fidelity vehicle dynamic model. Simulation results demonstrate the method's capability to utilize delayed outputs for precise future state prediction. Additionally, real-world testing further confirms the method's efficacy in mitigating dominant bi-directional delays, surpassing conventional controllers.
Given the validation via high-fidelity simulations and field tests, our approach holds genuine potential as a delay compensation control strategy within the future implementation of the ICV cloud control system.

\section*{Acknowledgments}
This work was supported by National Natural Science Foundation of China (52072212) and Joint Laboratory for Internet of Vehicle, Ministry of Education-China Mobile Communications Corporation.
\newpage
{\appendices
\section{Proof of Theorem \ref{thm:int_cont_sys_rep}}\label{apx:1}

If the input delay and output delay are all measurable as the assumptions in \cite{gonzalez2019gain}, the predictor-observer could be constructed as:
\begin{align}
    {\hat{Z}}(k+1)&=A{\hat{Z}}(k)+Fu(k)+L\left(A^{d_k^O}{\bar{y}}(k)-C{\hat{Z}}(k)\right), \\
    {\bar{y}}(k) &=y(k)+CA^{{-d}_k^O}\left(\Phi_k\left(h_1^I\right)+\Phi_k\left(h_2^I\right)+\Omega_k\left(d_k^O\right)\right) \label{y_pre_def_thm}, 
    %y_k&=Cx_{k-d_k^O},
\end{align}
\begin{equation}
    \Omega_k\left(d_k^O\right)=\sum_{i=0}^{d_k^O-1}A^{d_k^O-i-1}Bu(k-d_k^O+i-d_{k-d_k^O+i}^I) \label{omega_def_in_thm_prf}.
\end{equation}
(\ref{y_pre_def_thm}) is the core prediction function for the method in \cite{gonzalez2019gain}, which makes an ingenious connection between output delay compensation and input delay compensation. Nevertheless, since the controller is assumed to have no measurement on the accurate input delay (except the boundary information), then $d_{k-d_k^O+i}^I$ is not available to the predictor, and the prediction could not be generated directly as (\ref{omega_def_in_thm_prf}).  Therefore, constructing an estimation on $\Omega_k\left(d_k^O\right)$ as:
 \begin{align}
     {\bar{\mathrm{\Omega}}}_k\left(d_k^O\right)=\frac{1}{2}\sum_{i=0}^{d_k^O-1}A^{d_k^O-i-1}Bu(k-d_k^O+i-h_1^I) +\frac{1}{2}\sum_{i=0}^{d_k^O-1}A^{d_k^O-i-1}Bu(k-d_k^O+i-h_2^I).   
     \label{eq:omg_est_apdx}
\end{align} 
Here, from result in \cite{pan2023cloud}, to model the estimation error brought by the $\bar{\Omega }_k\left(d_k^O\right)$, introduce the input term $w_d(k)$, which is defined as:
\begin{equation}
    u(k-d_k^I)=\frac{1}{2}\sum_{j=1}^{2}u(k-h_j^I)+\frac{\tau}{2}w_{d}(k), \label{u_k-dik_1}
\end{equation}
where $\tau=h_2^I-h_1^I$ and:
\begin{align}
 w_{d}(k)&=\Delta_{d}(k)v(k) \label{w_dk_final},\\
\left\|\Delta_{d}(k)\right\|_{\infty} &\le 1.    \nonumber
\end{align}
Then, the estimation error could be indicated as:
\begin{equation}
    \Omega_k\left(d_k^O\right)-\bar{\Omega }_k\left(d_k^O\right)=\sum_{i=0}^{d_k^O-1}{A^{d_k^O-i-1}B}w_{d}(k-d_k^O+i) \label{omega_diff}
\end{equation}
Then, the estimated output ${\bar{y}}(k)$ and estimated (transformed) system states ${\bar{Z}}(k)$ could be constructed as:
\begin{align}
{\bar{y}}(k)&=CA^{{-d}_k^O}\left(A^{d_k^O}x(k-d_k^O)+\Phi_k\left(h_1^I\right)+\Phi_k\left(h_2^I\right)+{\bar{\mathrm{\Omega}}}_k\left(d_k^O\right)\right), \\
    {\bar{Z}}(k)&=A^{d_k^O}x(k-d_k^O)+{\bar{\mathrm{\Omega}}}_k\left(d_k^O\right)+\Phi_k\left(h_1^I\right)+\Phi_k\left(h_2^I\right). 
\end{align}
Hence, the dynamics of ${\hat{Z}}(k)$ could be indicated as:
\begin{equation}
    {\hat{Z}}(k+1)=A{\hat{Z}}(k)+Fu(k)+LA^{d_k^O}CA^{{-d}_k^O}\left({\bar{Z}}(k)-{\hat{Z}}(k)\right). \label{z_obs_mid_1}
\end{equation}
Define the observation error as:
\begin{equation}
    e(k)=Z(k)-{\hat{Z}}(k) \label{e_k_def}.
\end{equation}
Since the prediction could only be conducted with an estimated representation of $Z(k)$, define the estimation error as:
\begin{equation}
    {ep}(k)=Z(k)-{\bar{Z}}(k) \label{ep_k_def} 
\end{equation}
Here, the construction of the prediction error is discussed. The prediction error considered in this manuscript sourced from two major parts, the first part comes from the estimation ${\bar{\mathrm{\Omega}}}_k$, the second part comes from the model uncertainty, along with the mixing effection part. Besides, based on the definition of $Z(k)$, it reasonable to neglect the prediction error subject to input delay prediction and focus on only the state prediction error, which could be indicated as:
\begin{equation}
    ep(k)=x(k)-{\bar{x}}(k). \label{ep_k_def_2}
\end{equation}
To present the dynamic of $ep(k)$, based on the prediction equation subject to the output delay compensation and the forward system dynamic, it could be conducted that:
\begin{align}
    x(k)=&(A+\Delta\widetilde{A}(k))^{d_k^o}x(k-d_k^O)+\sum_{i=0}^{d_k^O-1}{(A+\Delta\widetilde{A}(k))^{d_k^O-i-1}(B+\Delta\widetilde{B}(k))u(k-d_k^O+i-d_{k-d_k^O+i}^I)} \label{real_sys_evo_org} \\
    {\bar{x}}(k)=&A^{d_k^O}x(k-d_k^O)+\frac{1}{2}\sum_{i=0}^{d_k^O-1}A^{d_k^O-i-1}Bu(k-d_k^O+i-h_1^I)+\frac{1}{2}\sum_{i=0}^{d_k^O-1}A^{d_k^O-i-1}Bu(k-d_k^O+i-h_2^I) \label{state_est_fcn}
\end{align}
where $\Delta\widetilde{A}(k)$ and $\Delta\widetilde{B}(k)$ are averaged equivalent representations of the uncertainty $\Delta A(k)$ and $\Delta B(k)$ through $\left[k-d_k^o,k\right]$, and they are assumed to have the same form of boundary, which are indicated as:
\begin{align}
\left(\Delta\widetilde{A}(k),\Delta\widetilde{B}(k)\right)&=\widetilde{\gamma}\widetilde{E}\Delta _p(k)\left(H_A,H_B\right), \label{avg_uc_rep} \\
\left\| \Delta _p(k) \right\|_{\infty} &\le 1 , \nonumber
\end{align}
where $\widetilde{E}$ is an averaged representation of $E$. To simplify the calculation, $\Delta\widetilde{A}(k)$ and $\Delta\widetilde{B}(k)$ are assumed to have the same order of infinitesimal and are far smaller than $A$ and $B$. Therefore, all terms with order higher than or equal to ${\Delta\widetilde{A}(k)}^2$ are neglected. Besides, in order to decouple the error brought by the estimation of input delay and the model uncertainty, it is also assumed that the error brought by the estimation of the input delay is with the same order of infinitesimal as the model uncertainty error. Then, (\ref{real_sys_evo_org}) will be further simplified as:
\begin{align}
    x\left(k\right)=&A^{d_k^o}x(k-d_k^o)+\sum_{i=0}^{d_k^o-1}A^{d_k^o-i-1}Bu(k-d_k^o+i)+C_{d_k^o}^1A^{d_k^o-1}\Delta\widetilde{A}(k)x(k-d_k^o) \nonumber \\
    & +\sum_{i=0}^{d_k^o-1}( A^{d_k^o-i-1}\Delta\widetilde{B}(k)+C_{d_k^o-i-1}^1A^{d_k^o-i-2}\Delta\widetilde{A}(k)B ) u(k-d_k^O+i-d_{k-d_k^O+i}^I),\label{real_sys_evo_sim} 
\end{align}
where $C_n^1$ is the binomial coefficient function. 
Hence, based on (\ref{omega_diff}),(\ref{state_est_fcn}), and (\ref{real_sys_evo_sim}), (\ref{ep_k_def_2})  could be indicates as:
\begin{align}
	 ep(k)=&C_{d_k^O}^1A^{d_k^O-1}\Delta \widetilde{A}(k) x(k-d_k^O)+\sum_{i=0}^{d_k^O-1}{A^{d_k^O-i-1}B}w_{d}(k-d_k^O+i) \nonumber \\ 
  &+\frac{1}{2}\sum_{i=0}^{d_k^O-1}{\left(A^{d_k^O-i-1}\Delta \widetilde{B}(k)+C_{d_k^O}^1A^{d_k^O-i-2}\Delta \widetilde{A}(k) B\right)u(k-d_k^O+i-h_1^I)} \nonumber \\
&+\frac{1}{2}\sum_{i=0}^{d_k^O-1}\left(A^{d_k^O-i-1}\Delta\widetilde{B}(k)+C_{d_k^O}^1A^{d_k^O-i-2}\Delta \widetilde{A}(k) B\right)u(k-d_k^O+i-h_2^I), 
\end{align}
corresponding to the estimation error, constructing an extra input term $w_p(k)$ and output term $y_p(k)$, which are indicated as:
\begin{align}
    {ep}(k)=&\widetilde{\gamma} \widetilde{E} w_{p}(k) ,\label{ep_k=Ew_p}\\
    w_{p}(k)=&\Delta_p(k)y_p(k), \left\|\Delta _p(k)\right\|_{\infty}\leq 1.
\end{align}
Thus, the output $y_p(k)$ could be indicated as:
\begin{align}
    &y_{p}(k)=C_{d_k^O}^1A^{d_k^O-1}H_Ax(k-d_k^O)+\sum_{i=0}^{d_k^O-1}{A^{d_k^O-i-1}B}w_{d}(k-d_k^O+i) \nonumber \\
    &+\frac{1}{2}\sum_{i=0}^{d_k^O-1}{\left(A^{d_k^O-i-1}H_B+C_{d_k^O}^1A^{d_k^O-i-2}H_A B\right)u(k-d_k^O+i-h_1^I)} \nonumber \\
    &+\frac{1}{2}\sum_{i=0}^{d_k^O-1}\left(A^{d_k^O-i-1}H_B+C_{d_k^O}^1A^{d_k^O-i-2}H_A B\right)u(k-d_k^O+i-h_2^I). \label{ypk_mid}
\end{align}
Here, it need to note that an assumption have been made to conduct (\ref{ypk_mid}), which is the $\widetilde{\Delta}(k)$ and the  $\widetilde{E}$ only indicate that amplitude of the model uncertainty, whereas $H_A$ and $H_B$ indicate the structure of the model uncertainty.
Now, reconsidering the regular system dynamic, and pulling out model uncertainty as:
\begin{align}
    x(k+1)=&\left(A+\Delta A\right)x(k)+\left(B+\Delta B\right)u(k-d_k^I) \nonumber \\
    =&Ax(k)+Bu(k-d_k^I)+\gamma Ew_{\Delta}(k),
\end{align}
where:
\begin{align}
    w_{\Delta}(k)=&\Delta (k)y_{\Delta}(k), \\
    y_{\Delta}(k)=&H_Ax(k)+H_B u(k-d_k^I). \label{y_deltak_1}
\end{align}
Since the transformed system states $Z(k)$ is defined as:
\begin{equation}
    Z(k)=x(k)+\phi_k\left(h_1^I\right)+\phi_k(h_2^I), \label{z_k_def}
\end{equation}
the one-step evolution of $Z$ could be indicated as:
\begin{equation}
Z(k+1)=x(k+1)+\phi_{k+1}\left(h_1^I\right)+\phi_{k+1}\left(h_2^I\right). \label{z_k+1}
\end{equation}
With (\ref{z_k_def}) and (\ref{z_k+1}), it could be conducted that:
\begin{align}
    Z(k+1)=&A\left(Z(k)-\phi_k\left(h_1^I\right)-\phi_k\left(h_2^I\right)\right)+Bu(k-d_k^I)+\gamma Ew_{\Delta}(k)+\phi_{k+1}\left(h_1^I\right)+\phi_{k+1}\left(h_2^I\right)  \nonumber \\
    =&AZ(k)-A\sum_{j=1}^{2}{\phi_k\left(h_j^I\right)}+Bu(k-d_k^I)+\sum_{j=1}^{2}{\phi_{k+1}\left(h_j^I\right)}+\gamma Ew_{\Delta}(k).
\end{align}
Meanwhile, with (\ref{y_deltak_1}) and (\ref{z_k_def}), output term $y_{\Delta}(k)$ could be indicated as:
\begin{align}
    y_{\Delta}(k)=&H_A\left(Z(k)-\phi_k\left(h_1^I\right)-\phi_k\left(h_2^I\right)\right)+H_Bu(k-d_k^I) \nonumber \\
    =&H_AZ(k)-H_A\sum_{j=1}^{2}{\phi_k\left(h_j^I\right)}+H_Bu(k-d_k^I). \label{y_delta_k_mid1}
\end{align}
Next, to pull out the delayed input term $u(k-d_k^I)$, based on (\ref{u_k-dik_1}) and (\ref{w_dk_final}), the dynamic of $Z$ and $y_{\Delta}$ could be indicated as:
\begin{align}
    Z(k+1)=&AZ(k)-A\sum_{j=1}^{2}{\phi_k\left(h_j^I\right)}+B\left(\frac{1}{2}\sum_{j=1}^{2}u(k-h_j^I)+\frac{\tau}{2}w_{d}(k)\right)+\sum_{j=1}^{2}{\phi_{k+1}\left(h_j^I\right)}+\gamma Ew_{\Delta}(k) \nonumber\\
    =& AZ(k)-A\sum_{j=1}^{2}{\phi_k\left(h_j^I\right)}+\frac{1}{2}B\sum_{j=1}^{2}u(k-h_j^I)+\frac{\tau}{2}Bw_{d}(k)+\sum_{j=1}^{2}{\phi_{k+1}\left(h_j^I\right)}+\gamma Ew_{\Delta}(k)  \\
    y_{\Delta}(k)=&H_AZ(k)-H_A\sum_{j=1}^{2}{\phi_k\left(h_j^I\right)}+H_B\left(\frac{1}{2}\sum_{j=1}^{2}u(k-h_j^I)+\frac{\tau}{2}w_{d}(k)\right).
\end{align}
Based on the definition of the $\phi_k\left(h_j^I\right)$, it could be conducted that:
\begin{equation}
    \phi_{k+1}\left(h_j^I\right)-{A\phi}_k\left(h_j^I\right)=A^{-h_j^I}\frac{B}{2}u(k)-\frac{B}{2}u(k-h_j^I),\ j=1,2
\end{equation}
Then, the dynamic of $Z$ could be simplify as:
\begin{align}
    Z(k+1)=&AZ(k)+\sum_{j=1}^{2}{A^{{-h}_j^I}\frac{B}{2}u(k)}+\frac{\tau}{2}Bw_{d}(k)+\gamma Ew_{\Delta}(k)  \nonumber\\
    =&AZ(k)+{Fu}(k)+\frac{\tau}{2}Bw_{d}(k)+\gamma Ew_{\Delta}(k), \label{z_dyn_mid_2}
\end{align}
where $F=\frac{A^{-h_1^I}+A^{{-h}_2^I}}{2}$. Since the controller is designed as:
\begin{equation}
    u(k)=K{\hat{Z}}(k)=K(Z(k)-e(k)), \label{u_k_final}
\end{equation}
thus, (\ref{z_dyn_mid_2}) could be indicated as:
\begin{align}
    Z(k+1)=&AZ(k)+FK(Z(k)-e(k))+\frac{\tau}{2}Bw_{d}(k)+\gamma Ew_{\Delta}(k) \nonumber \\
    =&\left(A+FK\right)Z(k)-\ FKe(k)+\frac{\tau}{2}Bw_{d}(k)+\gamma Ew_{\Delta}(k). \label{Z_k_dyn_final}
\end{align}
Since the predictor-observer generates the state estimation as:
\begin{equation}
    {\hat{Z}}(k+1)=A{\hat{Z}}(k)+Fu(k)+LA^{d_k^O}CA^{{-d}_k^O}\left({\bar{Z}}(k)-{\hat{Z}}(k)\right), \label{z_dyn_mid_3}
\end{equation}
and also considering the definition of the observation error (\ref{e_k_def}) and the estimation error (\ref{ep_k_def}), it could be conducted that:
\begin{equation}
    {\bar{Z}}(k)-{\hat{Z}}(k)=Z(k)-{\hat{Z}}(k)-\left(Z(k)-{\bar{Z}}(k)\right)=e(k)-{ep}(k),
\end{equation}
thus, (\ref{z_dyn_mid_3}) could be further indicated as:
\begin{align}
    {\hat{Z}}(k+1)=A{\hat{Z}}(k)+Fu(k)+LA^{d_k^O}CA^{{-d}_k^O}(e(k)-{ep}(k)),
\end{align}
thus, the dynamic of the observation error could be indicated accordingly as:
\begin{align}
    &e(k+1) \nonumber \\
    =&Z(k+1)-{\hat{Z}}(k+1) \nonumber \\
    =&\left(A+FK\right)Z(k)-\ FKe(k)+\frac{\tau}{2}Bw_{d}(k)+\gamma Ew_{\Delta}(k)-A{\hat{Z}}(k)-Fu(k)-LA^{d_k^O}CA^{{-d}_k^O}\left({\bar{Z}}(k)-{\hat{Z}}(k)\right) \nonumber \\
    =&A\left(Z(k)-{\hat{Z}}(k)\right)+\frac{\tau}{2}Bw_{d}(k)+\gamma Ew_{\Delta}(k)-LA^{d_k^O}CA^{{-d}_k^O}\left({\bar{Z}}(k)-{\hat{Z}}(k)\right) \nonumber \\
    =&Ae(k)+\frac{\tau}{2}Bw_{d}(k)+\gamma Ew_{\Delta}(k)-LA^{d_k^O}CA^{{-d}_k^O}\left(e(k)-{ep}(k)\right) \nonumber \\
    =&\left(A-LA^{d_k^O}CA^{{-d}_k^O}\right)e(k)+\frac{\tau}{2}Bw_{d}(k)+\gamma Ew_{\Delta}(k)+LA^{d_k^O}CA^{{-d}_k^O}\widetilde{\gamma}\widetilde{E}w_{p}(k), \label{e_k_dyn_final}
\end{align}
where the last step is conducted based on (\ref{ep_k=Ew_p}). Then, the problem comes as how to eliminate the delayed terms in $y_p(k)$. In the light of the same procedure to handle delayed input term, define the step difference of the system sates as:
\begin{equation}
    q(k)=x(k)-x(k-1). \label{q_k_def}
\end{equation}
Then, (\ref{ypk_mid}) could be further conducted as:
\begin{align}
    y_{p}(k)=&C_{d_k^O}^1A^{d_k^O-1}H_Ax(k-d_k^O)+\sum_{i=0}^{d_k^O-1}{A^{d_k^O-i-1}B}w_{d}(k-d_k^O+i) \nonumber \\
    &+\frac{1}{2}\sum_{i=0}^{d_k^O-1}{\left(A^{d_k^O-i-1}H_B+C_{d_k^O}^1A^{d_k^O-i-2}H_A B\right)u(k-d_k^O+i-h_1^I)} \nonumber \\
    &+\frac{1}{2}\sum_{i=0}^{d_k^O-1}\left(A^{d_k^O-i-1}H_B+C_{d_k^O}^1A^{d_k^O-i-2}H_A B\right)u(k-d_k^O+i-h_2^I) \nonumber\\
    =&C_{d_k^O}^1A^{d_k^O-1}H_A\left(x(k-1)-\sum_{f=1}^{d_k^O-1}q(k-f)\right)+\sum_{i=0}^{d_k^O-1}{A^{d_k^O-i-1}B}w_{d}(k-d_k^O+i) \nonumber \\
    &+\frac{1}{2}\sum_{i=0}^{d_k^O-1}\left(A^{d_k^O-i-1}H_B+C_{d_k^O}^1A^{d_k^O-i-2}H_A B\right)\left(u(k-1)-\sum_{f=1}^{d_k^O+h_1^I-i-1}v(k-f)\right) \nonumber \\
    &+\frac{1}{2}\sum_{i=0}^{d_k^O-1}\left(A^{d_k^O-i-1}H_B+C_{d_k^O}^1A^{d_k^O-i-2}H_A B\right)\left(u(k-1)-\sum_{f=1}^{d_k^O+h_2^I-i-1}v(k-f)\right). \label{y_pk_mid_4}
\end{align}
For representation convenience, define some operators as:
\begin{align}
    \beta_1=&C_{d_k^O}^1A^{d_k^O-1}, \\
    \beta_2=&\sum_{i=0}^{d_k^O-1}A^{d_k^O-i-1}, \\
    \beta_3=&\sum_{i=0}^{d_k^O-1}\left(C_{d_k^O}^1A^{d_k^O-i-2}\right),
\end{align}
then, (\ref{y_pk_mid_4}) could be further indicated as:
\begin{align}
    y_{p}(k)=&\beta_1H_Ax(k-1)-\beta_1H_A\sum_{f=1}^{d_k^O-1}q(k-f)+\sum_{i=0}^{d_k^O-1}{A^{d_k^O-i-1}B}w_{d}(k-d_k^O+i)+\left(\beta_2H_B+\beta_3H_A B\right)u(k-1) \nonumber \\
    &-\frac{1}{2}\sum_{i=0}^{d_k^O-1}\sum_{f=1}^{d_k^O+h_1^I-i-1}{(A^{d_k^O-i-1}H_B+C_{d_k^O}^1A^{d_k^O-i-2}H_A B)v(k-f)}\nonumber \\
    &-\frac{1}{2}\sum_{i=0}^{d_k^O-1}\sum_{f=1}^{d_k^O+h_2^I-i-1}{(A^{d_k^O-i-1}H_B+C_{d_k^O}^1A^{d_k^O-i-2}H_A B)v(k-f)}.\label{y_pk_mid_5}
\end{align}
Next, the delayed term in $y_{\Delta}(k)$ is handled with the same logic, then (\ref{y_delta_k_mid1}) could be furtherly indicated as:
\begin{align}
    &y_{\Delta}(k) \nonumber \\
    =&H_AZ(k)-H_A\sum_{j=1}^{2}{\phi_k\left(h_j^I\right)}+H_B(\frac{1}{2}\sum_{j=1}^{2}u(k-h_j^I)+\frac{\tau}{2}w_{d}(k)) \nonumber\\
    =&H_AZ(k)-H_A\frac{1}{2}\sum_{j=1}^{2}\sum_{i=0}^{h_j^I}{A^{-i-1}B(u(k-1)-\sum_{f=1}^{h_j^I-i-1}v(k-f))}+H_B(u(k-1)-\frac{1}{2}\sum_{j=1}^{2}\sum_{f=1}^{h_j^I-1}v(k-f)) \nonumber \\
    &+H_B\frac{\tau}{2}w_{d}(k) \nonumber \\
    =&H_AZ(k)-H_A(\Upsilon u(k-1)-\frac{1}{2}\sum_{j=1}^{2}\sum_{i=0}^{h_j^I}\sum_{f=1}^{h_j^I-i-1}{A^{-i-1}Bv(k-f)})+H_B(u(k-1)-\frac{1}{2}\sum_{j=1}^{2}\sum_{f=1}^{h_j^I-1}v(k-f)) \nonumber \\
    &+H_B\frac{\tau}{2}w_{d}(k), \label{y_delta_k_mid_2}
\end{align}
where:
\begin{equation}
    \Upsilon=\frac{1}{2}\sum_{j=1}^{2}\sum_{i=0}^{h_j^I}{A^{-i-1}B} .\label{upsilon_def}
\end{equation}
Inspired by operation in \cite{gonzalez2013robust,gonzalez2019robust}, taking the $\mathcal{Z}$-transform and inverse $\mathcal{Z}$-transform process to the delayed term (subject to input delay) in $y_{\Delta}(k)$ and introducing new input terms $w_{1}(k), w_{2}(k)$, (\ref{y_delta_k_mid_2}) could be furtherly indicated as:
\begin{equation}
    y_{\Delta}(k)=H_AZ(k)+\left(H_B-H_A\Upsilon\right)u(k-1)+\frac{\tau H_B}{2}w_{d}(k)-\mu_1H_Bw_{1}(k)+\mu_2H_Aw_{2}(k), \label{y_delta_k_final}
\end{equation}
where:
\begin{align}
    w_{1}(k)=&\mathcal{Z}^{-1}(\frac{1}{\mu_1}\mathcal{Z}(\frac{1}{2}\sum_{j=1}^{2}\sum_{f=1}^{h_j^I-1}v(k-f)))=\Delta _1(k)v(k) \\
    \mu _1=&\left\|\frac{1}{2}\sum_{j=1}^{2}\sum_{f=1}^{h_j^I-1}z^{-f}\right\|_{\infty}  \\
    w_{2}(k)=&\mathcal{Z}^{-1}(\frac{1}{\mu_2}\mathcal{Z}(\frac{1}{2}\sum_{j=1}^{2}\sum_{i=0}^{h_j^I}\sum_{f=1}^{h_j^I-i-1}{A^{-i-1}Bv(k-f)}))=\Delta _2(k)v(k) \\
    \mu _2=&\left\|\frac{1}{2}\sum_{j=1}^{2}\sum_{i=0}^{h_j^I}\sum_{f=1}^{h_j^I-i-1}{A^{-i-1}Bz^{-f}}\right\| _{\infty},
\end{align}
where $\Delta _i(i=1,2):v\rightarrow w_i$ are the corresponding operators with the unit norm ($\left\|\Delta _i\right\|_{\infty}=I$).
Following the same procedures and taking the same $\mathcal{Z}$-transform and inverse $\mathcal{Z}$-transform process to the delayed term in $y_p(k)$ and introducing new input terms $w_{3}(k), w_{4}(k),w_{5}(k),w_{6}(k),w_{7}(k)$, (\ref{y_pk_mid_5}) could be furtherly indicated as: 
\begin{align}
    y_{p}(k)=&\beta_1H_Ax(k-1)-\beta_1H_A\mu_7w_{7}(k)+\sum_{i=0}^{d_k^O-1}{A^{d_k^O-i-1}B}w_{d}(k-d_k^O+i)+\left(\beta_2H_B+\beta_3H_A B\right)u(k-1) \nonumber\\
    &-\frac{1}{2}\mu_3w_{3}(k)-\frac{1}{2}C_{d_k^O}^1\mu_4w_{4}(k)-\frac{1}{2}\mu_5w_{5}(k)-\frac{1}{2}C_{d_k^O}^1\mu_6w_{6}(k), \label{y_pk_mid_6}
\end{align}
where:
\begin{align}
    w_{3}(k)=&\mathcal{Z}^{-1}(\frac{1}{\mu_3}\mathcal{Z}(\sum_{i=0}^{d_k^O-1}\sum_{f=1}^{d_k^O+h_1^I-i-1}{A^{d_k^O-i-1}H_Bv(k-f)}))=\Delta _3 (k) v(k) \\
    \mu_3=&\left\|\sum_{i=0}^{d_k^O-1}\sum_{f=1}^{d_k^O+h_1^I-i-1}{A^{d_k^O-i-1}H_Bz^{-f}}\right\|_{\infty},\\
    w_{4}(k)=&\mathcal{Z}^{-1}(\frac{1}{\mu_4}\mathcal{Z}(\sum_{i=0}^{d_k^O-1}\sum_{f=1}^{d_k^O+h_1^I-i-1}{A^{d_k^O-i-2}H_A B v(k-f)})) =\Delta _4 (k) v(k), \\
    \mu_4=&\left\|\sum_{i=0}^{d_k^O-1}\sum_{f=1}^{d_k^O+h_1^I-i-1}{A^{d_k^O-i-2}H_A B z^{-f}}\right\|_{\infty}, \\
    w_{5}(k)=&\mathcal{Z}^{-1}(\frac{1}{\mu_5}\mathcal{Z}(\sum_{i=0}^{d_k^O-1}\sum_{f=1}^{d_k^O+h_2^I-i-1}{A^{d_k^O-i-1}H_Bv(k-f)}))=\Delta _5(k)v(k),\\
    \mu_5=&\left\|\sum_{i=0}^{d_k^O-1}\sum_{f=1}^{d_k^O+h_2^I-i-1}{A^{d_k^O-i-1}H_Bz^{-f}}\right\|_{\infty}, \\
    w_{6}(k)=&\mathcal{Z}^{-1}(\frac{1}{\mu_6}\mathcal{Z}(\sum_{i=0}^{d_k^O-1}\sum_{f=1}^{d_k^O+h_2^I-i-1}{A^{d_k^O-i-2}H_A B v(k-f)}))=\Delta _6 (k)v(k) ,\\
    \mu_6=& \left\|\sum_{i=0}^{d_k^O-1}\sum_{f=1}^{d_k^O+h_2^I-i-1}{A^{d_k^O-i-2}H_A B z^{-f}}\right\|_{\infty}, \\
    w_{7}(k)=&\mathcal{Z}^{-1}\left(\frac{1}{\mu_7}\mathcal{Z}\left(\sum_{f=0}^{d_k^O-1}{q(k-f)}\right)\right)=\Delta _7(k)q(k), \\
    \mu _7=&\left\|\sum_{f=0}^{d_k^O-1}z^{-f}\right\|_{\infty}
\end{align}
where $\Delta _i(i=3,\cdots,6):v\rightarrow w_i$ are also the corresponding operators with the unit norm ($\left\|\Delta _i\right\|_{\infty}=I$), and $\Delta _7:q\rightarrow w_7$ is also the corresponding operator with the unit norm ($\left\|\Delta _7\right\|_{\infty}=I$). After the operation, the only delayed term in (\ref{y_pk_mid_6}) is $w_{d}(k-d_k^O+i)$. Based on the previous proof, it could be conducted that $w_{d}(k-d_k^O+i)$ are bounded as:
\begin{equation}
    w_{d}(k-d_k^O+i)=\Delta_{d}(k-d_k^O+i)v(k).
\end{equation}
From the result in \cite{pan2023cloud}, a single operator ${\widetilde{\Delta}}_{d}(k)$ could be constructed to bound the group of the operators $\Delta_{d}(k-d_k^O+i)$, which could be indicated as:
\begin{equation}
    \left\|\Delta_{d}(k-d_k^O+i)\right\|_{\infty}\leq \left\|{\widetilde{\Delta}}_{d}(k)\right\|_{\infty}\leq 1 .
\end{equation}
Then, to present a compact representation of the equations, define the following output and input term as:
\begin{align}
    {\widetilde{v}}(k)=&\sum_{i=0}^{d_k^O-1}{A^{d_k^O-i-1}B}v(k), \\
   \widetilde{w}_d(k)=& \sum_{i=0}^{d_k^O-1}{A^{d_k^O-i-1}B}w_{d}(k-d_k^O+i),
\end{align}
and,
\begin{equation}
    \widetilde{w}_d(k) ={\widetilde{\Delta}}_{d}(k){\widetilde{v}}(k).
\end{equation}
Thus,  (\ref{y_pk_mid_6}) could be further conducted as:
\begin{align}
    y_{p}(k)=&\beta_1H_Ax(k-1)-\beta_1H_A\mu_7w_{7}(k)+\widetilde{w}_d(k)+\left(\beta_2H_B+\beta_3H_A B\right)u(k-1)-\frac{1}{2}\mu_3w_{3}(k)-\frac{1}{2}C_{d_k^O}^1\mu_4w_{4}(k) \nonumber\\
    &-\frac{1}{2}\mu_5w_{5}(k)-\frac{1}{2}C_{d_k^O}^1\mu_6w_{6}(k), \label{y_pk_final}
\end{align}
Besides, to present a compact dynamic of the augmented system, with (\ref{z_k_def}) and (\ref{y_delta_k_mid_2}), it could be conducted that:
\begin{align}
    x(k)=&Z(k)-\Upsilon u(k-1)+{\mu_2}w_{2}(k), \label{x_k_final} \\
    q(k)=&Z(k)-\Upsilon u(k-1)+{\mu_2}w_{2}(k)-x(k-1). \label{q_k_final}
\end{align}
Finally, with (\ref{Z_k_dyn_final})(\ref{x_k_final})(\ref{y_delta_k_final})(\ref{y_pk_final})(\ref{e_k_dyn_final}), the inter-connected system dynamic could be indicated as:
\begin{align}
   \left[ \begin{smallmatrix}
        Z(k+1)\\x(k)\\u(k)\\e(k+1)
    \end{smallmatrix}\right]=\left[\begin{smallmatrix}
        A+FK&0&0&-FK\\I&0&-\Upsilon &0\\ K&0&0&-K\\0&0&0&A-LA^{d_k^O}CA^{{-d}_k^O} \end{smallmatrix}\right]\left[\begin{smallmatrix}
            Z(k)\\x(k-1)\\u(k-1)\\e(k)\end{smallmatrix}\right] +\left[\begin{smallmatrix}
                \gamma E&\frac{\tau}{2}B&0&0&0&0&0&0&0&0&0\\
                0&0&0&0&0&\mu_2 & 0&0&0&0&0\\
                0&0&0&0&0&0&0&0&0&0&0\\
                \gamma E & \frac{\tau}{2}B & 0 &LA^{d_k^O}CA^{{-d}_k^O}\widetilde{\gamma}\widetilde{E}&0&0&0&0&0&0&0
            \end{smallmatrix}\right]\left[\begin{smallmatrix}
             w_{\Delta}(k)\\w_d(k)\\ \widetilde{w}_d(k)\\w_p(k)\\w_1(k)\\w_2(k)\\w_3(k)\\w_4(k)\\w_5(k)\\w_6(k)\\w_7(k)
            \end{smallmatrix}\right]
\end{align}
\begin{align}
    \left[\begin{smallmatrix}
        y_{\Delta}(k)\\v(k)\\ \widetilde{v}(k)\\y_p(k)\\v(k)\\v(k)\\v(k)\\v(k)\\v(k)\\v(k)\\q(k)\\
    \end{smallmatrix}\right]=&
    \left[\begin{smallmatrix}
        H_A&0&H_B-H_A\Upsilon&0\\
        K&0&-I&-K\\\beta_2BK&0&-\beta_2B&-\beta_2BK\\
        0&\beta_1H_A&\beta_2H_B+\beta_3H_AB&0\\ 
        K&0&-I&-K\\K&0&-I&-K\\K&0&-I&-K\\K&0&-I&-K\\K&0&-I&-K\\K&0&-I&-K\\I&-I&-\Upsilon &0
    \end{smallmatrix}\right]
    \left[\begin{smallmatrix}
        Z(k)\\x(k-1)\\u(k-1)\\e(k)
    \end{smallmatrix}\right]  \nonumber\\
    &+
    \left[\begin{smallmatrix}
        0&\frac{\tau H_B}{2} &0&0&-\mu_1H_B & \mu_2H_A &0&0&0&0&0\\
        0&0&0&0&0&0&0&0&0&0&0\\
        0&0&0&0&0&0&0&0&0&0&0\\
        0&0&I&0&0&0&-\frac{1}{2}\mu_3&-\frac{1}{2}C_{d_k^O}^1\mu_4&-\frac{1}{2}\mu_5&-\frac{1}{2}C_{d_k^O}^1\mu_6&-\mu_7\beta_1H_A\\
        0&0&0&0&0&0&0&0&0&0&0\\
        0&0&0&0&0&0&0&0&0&0&0\\
        0&0&0&0&0&0&0&0&0&0&0\\
        0&0&0&0&0&0&0&0&0&0&0\\
        0&0&0&0&0&0&0&0&0&0&0\\
        0&0&0&0&0&0&0&0&0&0&0\\
        0&0&0&0&0&\mu_2&0&0&0&0&0\\
    \end{smallmatrix}\right]
    \left[\begin{smallmatrix}
             w_{\Delta}(k)\\w_d(k)\\ \widetilde{w}_d(k)\\w_p(k)\\w_1(k)\\w_2(k)\\w_3(k)\\w_4(k)\\w_5(k)\\w_6(k)\\w_7(k)
            \end{smallmatrix}\right],
\end{align}
\begin{align}
    \left[\begin{smallmatrix}
             w_{\Delta}(k)\\w_d(k)\\ \widetilde{w}_d(k)\\w_p(k)\\w_1(k)\\w_2(k)\\w_3(k)\\w_4(k)\\w_5(k)\\w_6(k)\\w_7(k)
            \end{smallmatrix}\right]=
\left[
    \begin{smallmatrix}
        \Delta(k) &&&&&&&&&& \\&\Delta _d(k) &&&&&&&&&\\&&\widetilde{\Delta}_d(k) &&&&&&&&\\&&&\Delta_p(k)&&&&&&&\\&&&&\Delta_1(k) &&&&&&\\&&&&&\Delta_2(k)&&&&&\\&&&&&&\Delta_3(k)&&&&\\&&&&&&&\Delta_4(k)&&&\\&&&&&&&&\Delta_5(k)&&\\&&&&&&&&&\Delta_6(k)&\\&&&&&&&&&&\Delta_7(k)
    \end{smallmatrix}\right]
    \left[\begin{smallmatrix}
        y_{\Delta}(k)\\v(k)\\ \widetilde{v}(k)\\y_p(k)\\v(k)\\v(k)\\v(k)\\v(k)\\v(k)\\v(k)\\q(k)\\
    \end{smallmatrix}\right],
\end{align}
which completes the proof.
}
\newpage
%Bibliography
\bibliographystyle{unsrt}

\end{document}